\documentclass{article}
\usepackage{amssymb,amsmath,multirow}
\usepackage[dvips]{graphicx}
\usepackage{epsfig}
\usepackage{graphics}
\newtheorem{theorem}{Theorem}
\newtheorem{lemma}{Lemma}
\newtheorem{remark}{Remark}
\newtheorem{defi}{Definition}

\begin{document}

\title{On RIC bounds of Compressed Sensing Matrices for Approximating Sparse
Solutions Using $\ell_q$ Quasi Norms
\thanks{This research was partially supported by Taiwan National
Science Council under NSC 101-2115-M-006-005; by Taiwan National
Center for Theoretic Studies (South); by National Natural Science
Foundation of China under grant 11001006 and 91130019/A011702, and
by the fund of State Key Laboratory of Software Development
Environment under grant SKLSDE-2011ZX-15.} }
\author{\small Yong Hsia \\
{\small State Key Laboratory of Software Development
              Environment; LMIB of the Ministry of Education; }\\
{\small School of Mathematics and System Sciences, Beihang University, Beijing 100191, P. R. China},\\
\small  Ruey-Lin Sheu \\
{\small Department of Mathematics, National Cheng Kung University, Taiwan}
  }

\date{}

\maketitle
\begin{abstract}
This paper follows the recent discussion on the sparse solution
recovery with quasi-norms $\ell_q,~q\in(0,1)$ when the sensing
matrix possesses a Restricted Isometry Constant $\delta_{2k}$ (RIC).
Our key tool is an improvement on a version of ``the converse of a
generalized Cauchy-Schwarz inequality'' extended to the setting of
quasi-norm. We show that, if $\delta_{2k}\le 1/2$, any minimizer of
the $l_q$ minimization, at least for those $q\in(0,0.9181]$, is the
sparse solution of the corresponding underdetermined linear system.
Moreover, if $\delta_{2k}\le0.4931$, the sparse solution can be
recovered by any $l_q, q\in(0,1)$ minimization. The values $0.9181$
and $0.4931$ improves those reported previously in the literature.

\end{abstract}


\section{Introduction}

Given a $m\times n$ matrix $\Phi$ with $m \ll n$ and a nonzero
vector $b\in \mathbb{R}^m$, one of the recently popular problem in
compressed sensing is to find the sparsest solution of the
underdetermined linear systems $\Phi x=b$. Here the sensing matrix
$\Phi$ is assumed to obey a Uniform Uncertainty Principle (UUP) for
every $k$ sparse vector $x$, or to possess a Restricted Isometry
Constant $\delta_{2k}$ (RIC) defined in \cite{C05} as follows:
\begin{defi}
For $k=1,2,\cdots$, the restricted isometry constant is the smallest
number $\delta_k$ such that
\[
(1-\delta_k)\|x\|_2^2\le\|\Phi x\|_2^2 \le (1+\delta_k)\|x\|_2^2
\]
holds for all $k$-sparse vector $x\in \mathbb{R}^n$ with $\|x\|_0\le
k$, where $\|x\|_0$ denotes the number of nonzero elements of $x$.
\end{defi}
For simplicity, we only discuss the sparse solution for noiseless
recovery, as it can be easily extended to the noisy recovery case.
See for example \cite{C08,Lai09}.

As the sparsity is fundamental in signal processing, the problem can
be mathematically formulated as to find
\begin{equation}
x^*=arg\min \{\|x\|_0\mid~  \Phi x=b, x\in \mathbb{R}^n\}.
\label{sp}
\end{equation}
Unfortunately, the formulation (\ref{sp}) is practically intractable
due to its NP-hardness \cite{N95}. For more information on the issue
and related applications in signal and image processing, see
\cite{B09} and references therein.

A common alternative is to consider the following convex problem
using the $\ell_1$ norm
\begin{equation}
x^{(1)}=arg\min \{\|x\|_1\mid~  \Phi x=b, x\in \mathbb{R}^n \}.
\label{lq:1}
\end{equation}
Let $x^*$ be the sparse solution to (\ref{sp}). Define
\begin{equation}
T_0=\{i\in\{1,2,\cdots,n\} \mid ~x^*_i\neq 0\},\label{T0}
\end{equation}
and define $x_{T_0}$ for any $x\in \mathbb{R}^n$ as
\begin{equation}
(x_{T_0})_i=\left\{\begin{array}{cl}x_i,& {\rm if}~i\in T_0,
\\0,&{\rm otherwise}.
\end{array} \right. \label{xT0}
\end{equation}
Then, a new $\ell_1$ recovery result, which extended many existing
ones in literature such as \cite{C08,Lai09,C100,F10}, is stated
below.
\begin{theorem} [\cite{Li11}] Suppose
$\delta_{2k}<\frac{77-\sqrt{1337}}{82}\approx 0.4931$, the solution
$x^{(1)}$ of the problem (\ref{lq:1}) satisfies
\[
\|x-x^{(1)}\|_2\le C_0k^{-1/2}\|x-x_{T_0}\|_1,~~\forall x\in\{x\mid
\Phi x=b\},
\]
where $C_0$ is a positive constant dependent on $\delta_{2k}$. It
follows that  if $x$ is k-sparse, the recovery is exact.
\end{theorem}

On the other hand, nonconvex $\ell_q$ quasi-norm minimization with
$q\in(0,1)$ is also considered to recover the sparse solution
\cite{G03}. It is to solve, for a number $q\in(0,1)$,
\begin{equation}
x^{(q)}=arg\min \{\|x\|_q \mid~  \Phi x=b, x\in \mathbb{R}^n \}.
\label{lq}
\end{equation}
More studies on the nonconvex $\ell_q$ minimization problem can be
found in \cite{Ch07,Lai09,Lai11,D09}. In particular, the following
two theorems will be strengthened in this paper.

\begin{theorem}[\cite{Lai09}]\label{thm:09}
Suppose $\delta_{2k}<2(3-\sqrt{2})\approx 0.4531$. Then for any $q\in(0,1]$
\[
\|x-x^{(q)}\|_q\le C_0\|x-x_{T_0}\|_q,~~\forall x\in\{x\mid \Phi
x=b\},
\]
where $C_0$ is a positive constant dependent on $\delta_{2k}$. In
particular, if $x=x^*$ is k-sparse, the recovery is exact.
\end{theorem}

\begin{theorem}[\cite{Lai11}]\label{thm:11}
Suppose $\delta_{2k}<1/2$. There exists a number $q_0\in(0,1]$ such
that for any $q<q_0$, each minimizer $x^q$ of the $\ell_q$
minimization (\ref{lq}) is the sparse solution of (\ref{lq}).
Furthermore, there exists a positive constant $C_q$ dependent on $q$
and $\delta_{2k}$ such that
\[
\|x-x^{(q)}\|_q\le C_q\|x-x_{T_0}\|_q,~~\forall x\in\{x\mid \Phi
x=b\}.
\]
\end{theorem}
More specifically, the upper bound of $\delta_{2k}$ in Theorem
\ref{thm:09} can be improved to $0.4931$ in Theorem \ref{sp} of the
paper. Secondly, the sufficient condition $\delta_{2k}<1/2$ in
Theorem \ref{thm:11} can be extended to $\delta_{2k}\le 1/2$ and the
threshold $q_0$ can be precisely estimated to be at least $0.9181$
by Theorem \ref{thm:2}. This is indeed a very surprising result
since $q_0$, if it would have been computed by the analysis in
\cite{Lai11}, is only $0.0513$. Our main tool to achieve these
results is an improvement on a version of ``the converse of a
generalized Cauchy-Schwarz inequality'' extended to the setting of
$\ell_q$ quasi-norms. The key inequality is stated and proved in
Section 2 below.

\section{A Key Inequality}

According to Cauchy-Schwarz inequality, we have the standard
inequality
\[
\|x\|_2\ge \frac{\|x\|_1}{\sqrt{n}}.
\]
The following converse of the above inequality is very recent:
\begin{lemma}[\cite{Cai10}]\label{lem:00}
For any $x\in \mathbb{R}^n$,
\begin{equation*}
\|x\|_2\le
\frac{\|x\|_1}{\sqrt{n}}+\frac{\sqrt{n}}{4}\left(\max_{1\leq i\le
n}|x_i|-\min_{1\leq i\le n}|x_i|\right).
\end{equation*}
\end{lemma}
On the other hand, Cauchy-Schwarz inequality can be extended to the
setting of quasi-norm $\|x\|_q,~q\in(0,1)$ with
\begin{equation}
\|x\|_2\ge \frac{\|x\|_q}{n^{1/q-1/2}}\label{quasi-norm:1}
\end{equation}
by using  H$\rm \ddot{o}$lder's inequality. The first converse of
(\ref{quasi-norm:1}) was proposed in (\cite{Lai11}).
\begin{lemma}[\cite{Lai11}]\label{lem:01}
Fix $0<q<1$. For any $x\in \mathbb{R}^n$,
\begin{equation}
\|x\|_2\le \frac{\|x\|_q}{ n^{1/q-1/2} }+ \sqrt{n} \left(\max_{1\leq
i\le n}|x_i|-\min_{1\leq i\le n}|x_i|\right).\label{Lsai-main}
\end{equation}
\end{lemma}
Our key result, Lemma \ref{lem:new} below, gives a sharpened
estimation on the right hand side of (\ref{Lsai-main}). When $p_q$
in (\ref{p_q}) is replaced by $1$ for any $q\in(0,1)$, Lemma
\ref{lem:new} reduces to Lemma \ref{lem:01}.

\begin{lemma}\label{lem:new}
For $q\in(0,1)$ and $x\in \mathbb{R}^n$, there is
\begin{equation}
\|x\|_2\le \frac{\|x\|_q}{ n^{1/q-1/2} }+
p_q\sqrt{n}\left(\max_{1\leq i\le n}|x_i|-\min_{1\leq i\le
n}|x_i|\right), \label{main:0}
\end{equation}
where
\begin{equation}
p_q:=\left(\frac{q}{2}\right)^{{\frac
{q}{2-q}}}-\left(\frac{q}{2}\right)^{{\frac {2}{2-q}}}.\label{p_q}
\end{equation}
Moreover, $p_q$ is a decreasing convex function of $q\in(0,1)$ with
\[
\lim_{q\rightarrow 0}p_q=1~~\hbox{and}~~\lim_{q\rightarrow1}p_q=1/4.
\]
\end{lemma}
\begin{proof}
Due to the symmetry of the inequality (\ref{main:0}) in the
components $\vert x_1\vert, \vert x_2\vert, \ldots, \vert x_n\vert$,
we only have to prove the case for $x\in S=\{(x_1,
x_2,\ldots,x_n)\not=0\vert\ x_1\ge x_2\ge \cdots \ge x_n\ge 0\}$
(notice that $x=0$ is a trivial case). Furthermore, suppose the
inequality (\ref{main:0}) is true for $x\in S, x_1=1$. By
substituting $\frac{x}{x_1}, x\in S$ into (\ref{main:0}) and
canceling the common factor $\frac{1}{x_1}$, we immediately
generalize the result to all $x\in S$. In other words, our goal is
to show
\begin{equation}
\|x\|_2\le \frac{\|x\|_q}{ n^{1/q-1/2} }+
p_q\sqrt{n}\left(1-x_n\right),\ x\in S_1=\{x\in S\vert\
x_1=1\}\label{main:1}
\end{equation}
where $p_q$ is a function of $q$ specified in (\ref{p_q}).

Following the approach in \cite{Lai11}, we define for any fixed
$q\in(0,1)$
\begin{equation*}
f(x)=\|x\|_2 -\frac{\|x\|_q}{ n^{1/q-1/2 }}
\end{equation*}
and compute the first order partial derivatives as
\begin{equation*}
\frac{\partial f(x)}{\partial
x_i}=\frac{x_i}{\|x\|_2}-\frac{\|x\|_q^{1-q}x_i^{q-1}}{ n^{1/q-1/2}
}. 
\end{equation*}
Note that, when $x_i$ increases with all other components fixed, the
following two terms
\begin{equation}
\frac{\|x\|_2}{x_i}=\sqrt{\sum_{j=1}^n\left(\frac{x_j}{x_i}\right)^2}\label{df-1}
\end{equation}
and
\begin{equation}
\frac{\|x\|_q^{1-q}x_i^{q-1}}{ n^{1/q-1/2} }=\frac{1}{ n^{1/q-1/2} }
\left(\sum_{j=1}^n(\frac{x_j}{x_i})^{q}\right)^{\frac{1-q}{q}}\label{df-2}
\end{equation}
are both decreasing. As the result, $\frac{\partial f(x)}{\partial
x_i}$ is increasing and $f(x)$ is convex in each of the components
$x_i$ for $i = 1,2,\cdots,n$. Analogously from (\ref{df-1}) and
(\ref{df-2}), we can show that the composite function $g: R^{n-1}
\rightarrow R$ such that
$$g(x_1,x_3,x_4,\ldots,x_n)=f(x_1,x_3,x_3,x_4,\ldots,x_n)$$ is also
convex in the variable $x_3$ while all other components $x_1,
x_4,\ldots,x_n$ remaining fixed. Likewise, we can conclude that
$f(1,\ldots,1,x_n,\ldots,x_n)$ is convex in the variable $x_n$ where
$x_n$ is repeated for a couple of times.

Since the maximum of a convex function always happens on the
boundary, we have
\begin{eqnarray}
&&\max\limits_{1\ge x_2\ge x_3\cdots\ge x_n} f(1, x_2, \ldots, x_n)\nonumber\\
&=&\max\limits_{1\ge x_3\ge\cdots\ge
x_n}\left\{\max\limits_{x_2\in[x_3, 1]}
f(1, x_2, \ldots, x_n)\right\}\nonumber\\
&=&\max\limits_{1\ge x_3\ge\cdots\ge
x_n}\max\left\{f(1,1,x_3,\ldots,x_n),
f(1,x_3,x_3,\ldots,x_n)\right\}\nonumber\\
&=&\max\left\{\max\limits_{1\ge x_3\ge\cdots\ge
x_n}f(1,1,x_3,\ldots,x_n), \max\limits_{1\ge x_3\ge\cdots\ge
x_n}f(1,x_3,x_3,x_4\ldots,x_n)\right\}.\label{df-3}
\end{eqnarray}
In (\ref{df-3}), since $f(1,x_3,x_3,x_4\ldots,x_n)$ is convex in
$x_3$, it follows that
$$\max\limits_{x_3\in[x_4,1]}f(1,x_3,x_3,x_4\ldots,x_n)=\max\{f(1,1,1,x_4,\ldots,x_n),
f(1,x_4,x_4,x_4,x_5,\ldots,x_n)\}.$$ Repeating the arguments
iteratively, we can thus express the maximum of $f$ only in terms of
$1$ and $x_n$ as follows: 
$$h(x_n)=\max\limits_{1\ge x_2\ge x_3\cdots\ge x_n} f(1, x_2, \ldots, x_n)
=f(1,\ldots,1,x_n,\ldots,x_n),~x_n\in[0,1].$$ Suppose the
distribution of $x_1=1$ appears for $r$ times $(1\le r \le n)$ in
the maximum solution of $f$, we have
\begin{equation*}
h(x_n)
=\sqrt{r(1-x_n^2)+nx_n^2}-\frac{(r(1-x_n^q)+nx_n^q)^{1/q}}{n^{1/q-1/2}}.
\end{equation*}
Since $h(x_n)$ is convex and $h(1)=0$, we have
\[
h(x_n)\le (1-x_n) h(0) + x_n h(1)=(1-x_n) h(0).
\]
Then it holds that
\begin{eqnarray}
{f(x)}&\le& {h(x_n)} \nonumber\\
&\le& (1-x_n)h(0) \nonumber\\
&\le&
(1-x_n)\max_{r\in\{1,2,\ldots,n\}} \left\{\sqrt{r( 1^2-0^2)+n0^2}
-\frac{(r(1^q-0^q)+n0^q)^{1/q}}{n^{1/q-1/2}}\right\}\label{fxk}\\
&\le&(1-x_n)\max_{r\in[1,n]} \left\{\sqrt{r}-\frac{r^{1/q}}{n^{1/q-1/2}}\right\}\label{main:1}\\
&=& (1-x_n)\left(\left(\frac{q}{2}\right)^{{\frac
{q}{2-q}}}-\left(\frac{q}{2}\right)^{{\frac
{2}{2-q}}}\right)\sqrt{n}.\nonumber
\end{eqnarray}
where (\ref{fxk}) is an upper bound estimation for $h(0)$ over the
unknown parameter $r$ (the number of times $x_1=1$ is repeated), and
(\ref{main:1}) is a concave maximization problem since $q\in(0,1)$
and $r$ is relaxed to a real number on $[1,n]$.

Finally, it is easy to verify that
\[
\lim_{q\rightarrow 0}p_q=1,~ \lim_{q\rightarrow1}p_q=1/4,
\]
and $p_q$ is a decreasing function of $q\in(0,1)$ since
\begin{eqnarray}
\frac{d}{dq}p_q&=& \frac{(q/2)^{\frac{q}{2-q} }}{(2-q)^2} ( 2\ln(q/2) +2-q)-
\frac{(q/2)^{\frac{2}{2-q} }}{(2-q)^2}( 2\ln(q/2) +\frac{4}{q}-2)\nonumber\\
&=&\frac{ 2\ln(q/2)}{(2-q)^2}p_q<0. \label{eq:2}
\end{eqnarray}
Moreover, the convexity of $p_q$ over $q\in(0,1)$ can be verified by
\begin{eqnarray*}
\frac{d^2}{dq^2}p_q = \frac{(2\ln(q/2)+2-q)^2+(2-q)^3/q}{(2-q)^4}p_q>0.
\end{eqnarray*} and
the proof is thus completed.
\end{proof}



In the following, we give an upper estimate of $p_q$ for small $q$,
which will be used later.
\begin{lemma}\label{lem:a1}
Denote Euler's number by $e=2.718\cdots$. It holds that
\[
p_q< 1+\frac{q\ln(q/2)}{2-q},~\forall
q\in(0,0.4797]\subset(0,1-\sqrt{2}/e].
\]
\end{lemma}
\begin{proof}
Since $e^x<1+x+x^2/2$ for $x< 0$, we have
\[
p_q=\left(q/2\right)^{{\frac {q}{2-q}}}-\left(q/2\right)^{{\frac
{2}{2-q}}}=e^{ \frac {q\ln(q/2)}{2-q} }-\left(q/2\right)^{{\frac
{2}{2-q}}}<1+ \frac {q\ln(q/2)}{2-q} + \frac
{q^2\ln^2(q/2)}{2(2-q)^2}-\left(q/2\right)^{{\frac {2}{2-q}}}.
\]
To prove the lemma, it is sufficient to show that
\begin{equation}
\frac {-q\ln(q/2)}{\sqrt{2}(2-q)}\le \left(q/2\right)^{{\frac
{1}{2-q}}},~\forall q\in(0,1-\sqrt{2}/e].\label{lemma4:1}
\end{equation}
Since $\frac {1}{2-q}\le\frac {1}{1+\sqrt{2}/e }$ and
$\left(q/2\right)^{{\frac {1}{2-q}}}\ge \left(q/2\right)^{{\frac
{1}{1+\sqrt{2}/e}}}$, the inequality (\ref{lemma4:1}) can be
confirmed by verifying
\[
\sqrt{2}\frac {-(q/2)\ln(q/2)}{(1+\sqrt{2}/e)}\le
\left(q/2\right)^{{\frac {1}{1+\sqrt{2}/e}}},~\forall
q\in(0,1-\sqrt{2}/e],
\]
or equivalently, by verifying
\[
 -(q/2)^{\frac {\sqrt{2}/e}{1+\sqrt{2}/e}}\ln(q/2) \le
\frac{1+\sqrt{2}/e}{\sqrt{2}},~\forall q\in(0,1-\sqrt{2}/e].
\]

Let $y=(q/2)^{\frac {\sqrt{2}/e}{1+\sqrt{2}/e}}\in(0,1)$. We then
have the desired result by
\[
 -(q/2)^{\frac {\sqrt{2}/e}{1+\sqrt{2}/e}}\ln(q/2)=-\frac {1+\sqrt{2}/e }{\sqrt{2}/e
 }y\ln(y)\le  \frac{1+\sqrt{2}/e}{\sqrt{2}}
\]
because the negative entropy function $-y\ln(y)$ attains the maximum
value of $1/e$.
\end{proof}

\section{Main Results}

Let Null($\Phi$) be the null space of $\Phi$; $x^*,x^{(q)}$ be the
solutions to (\ref{sp}) and (\ref{lq}), respectively; $T_0$ be
defined in (\ref{T0}). Suppose $\|x^*\|_0=k$ and define
\[
T_0^c=\{1,2,\cdots,n\} \setminus T_0.
\]
The following null space property is essential. However, a refined
version is immediately stated in Lemma \ref{lem:03}.
\begin{lemma}[\cite{G03}]\label{lem:02}
$x^{(q)}$ is the unique sparse solution  $x^*$ if and only if
\begin{equation}
\|h_{T_0}\|_q<\|h_{T_0^c}\|_q,~ \forall  h\in {\rm Null}(\Phi),~
h\neq 0, \label{h:1}
\end{equation}
where $h_{T_0}$ and $h_{T_0^c}$ are similarly defined as in
(\ref{xT0}).
\end{lemma}
\begin{lemma}\label{lem:03}
$x^{(q)}$ is the unique sparse solution  $x^*$ if and only if
\begin{equation}
\|h_{T_0}\|_q<\|h_{T_0^c}\|_q,~ \forall  h\in {\rm
Null}(\Phi),~h_{T_0^c}\neq 0.\label{h:2}
\end{equation}
\end{lemma}
\begin{proof}
It is sufficient to study the difference between (\ref{h:1}) and
(\ref{h:2}). Suppose $0\neq h\in {\rm Null}(\Phi)$ and $h_{T_0^c}=
0$.  It follows that $h_{T_0 }\neq  0$ and $\Phi h_{T_0}=0$.
Therefore,
\[
\Phi (x^*+th_{T_0})=b,~\forall t\in \mathbb{R},
\]
and
\[
\min_{t\in \mathbb{R}} \|x^*+th_{T_0}\|_0 \le k-1,
\]
which contradicts the optimality of $x^*$.
\end{proof}

The purpose of this research is to establish sufficient conditions
for (\ref{h:2}) with the help of Lemmas \ref{lem:new} and
\ref{lem:a1} so that
\[
\tau(h,q):=\frac{\|h_{T_0}\|_q}{\|h_{T_0^c}\|_q}<1,~\forall h\in
{\rm Null}(\Phi),~h_{T_0^c}\neq 0.
\]
To this end, let
\[
h= h_{T_0}+h_{T_1}+h_{T_2}+\cdots ,
\]
where $T_1$ corresponds to the locations of the $k$ largest entries
of $h_{T_0^c}$, $T_2$ the locations of the next $k$ largest entries
of $h_{T_0^c}$ and so on. Without loss of generality, we assume
\[
h=(h_{T_0},h_{T_1},h_{T_2},\cdots)^T
\]
with the cardinality of $T_i$ being equal to $k$ for
$i=0,1,2,\cdots$. Define a ratio
\[
t:=t(h,q)\in(0,1]
\]
such that
\[
\|h_{T_1}\|_q^q=t\sum_{i\ge 1}\|h_{T_i}\|_q^q.
\]
According to Lemma \ref{lem:03}, we only focus on nonzero
$h_{T_0^c}$, which immediately implies that $h_{T_1}\neq 0$, i.e.,
$t>0$. Several technique lemmas are needed.

\begin{lemma}[\cite{Lai11}] \label{lem:3}
For $q\in(0,1)$, we have
\begin{equation*}
\sum_{i\ge 2} \|h_{T_i}\|_2^2\le \frac{1}{k^{(2-q)/q}}(1-t)t^{(2-q)/q}
\left(\sum_{i\ge 1}\|h_{T_i}\|_q^q \right)^{2/q}
\end{equation*}
\end{lemma}

\begin{lemma} \label{lem:4}
For $q\in(0,1),~ h\in {\rm Null}(\Phi),~h_{T_0^c}\neq 0$, we have
\begin{equation*}
\sum_{i\ge 2} \|h_{T_i}\|_2 \le \frac{1+(p_q-1)t^{1/q}}{k^{1/q-1/2}}
\left(\sum_{i\ge 1}\|h_{T_i}\|_q^q \right)^{1/q}
\end{equation*}
\end{lemma}
\begin{proof}
We first apply Lemma \ref{lem:new} to each $h_{T_i}$ to get
\begin{equation}
k^{1/q-1/2}\|h_{T_i}\|_2\le \|h_{T_i}\|_q+p_qk^{1/q}
(|h_{ik+1}|-|h_{ik+k}|),~~i=2,3,\ldots \label{Lemma-008}
\end{equation}
and sum up (\ref{Lemma-008}) over all $i\ge 2$. Then, 
\begin{eqnarray*}
k^{1/q-1/2}\sum_{i\ge 2}\|h_{T_i}\|_2&\le& \sum_{i\ge 2}\|h_{T_i}\|_q
+ p_qk^{1/q}\left\{|h_{2k+1}|-(|h_{2k+k}|-|h_{3k+1}|)-\cdots\right\}\\
&\le& \sum_{i\ge 2}\|h_{T_i}\|_q +
p_qk^{1/q} |h_{2k+1}|\\
&=& \sum_{i\ge 2}\|h_{T_i}\|_q +
p_qk^{1/q} (|h_{2k+1}|^q)^{1/q}\\
&\le& \sum_{i\ge 2}\|h_{T_i}\|_q +
p_qk^{1/q} (\|h_{T_1}\|^q_q/k)^{1/q}\\
&=& \sum_{i\ge 1}\|h_{T_i}\|_q+(p_q-1)\left( \|h_{T_1}\|_q^q\right)^{1/q} \\
&\le& \left(\sum_{i\ge 1}\|h_{T_i}\|_q^q\right)^{1/q}+(p_q-1)
\left( t\sum_{i\ge1}\|h_{T_i}\|_q^q\right)^{1/q}~~(\hbox{since~}q<1)\\
&=&(1+(p_q-1)t^{1/q})\left(\sum_{i\ge 1}\|h_{T_i}\|_q^q\right)^{1/q}.
\end{eqnarray*}
\end{proof}

\begin{lemma} \label{lem:6}
For $q\in(0,1),~ h\in {\rm Null}(\Phi),~h_{T_0^c}\neq 0$, we have
\begin{equation*}
\|\Phi(h_{T_0}+h_{T_1})\|_2^2=\|\Phi(\sum_{j\ge2}h_{T_j})\|_2^2 \le
\left(\frac{(1-t)t^{(2-q)/q}}{k^{2/q-1}}+\frac{\delta_{2k}(1+(p_q-1)t^{1/q})^2}{k^{2/q-1}}\right)
\left(\sum_{i\ge 1}\|h_{T_i}\|_q^q\right)^{2/q}.
\end{equation*}
\end{lemma}
\begin{proof}
According to Lemma 1.2 in \cite{C05}, we have
\[
\langle\Phi(h_{T_i}),\Phi(h_{T_j})\rangle \le \delta_{2k}
\|h_{T_i}\|_2\|h_{T_j}\|_2.
\]
Therefore,
\begin{eqnarray*}
\|\Phi(\sum_{j\ge2}h_{T_j})\|_2^2&=&\sum_{i,j\ge2}\langle\Phi(h_{T_i}),\Phi(h_{T_j})\rangle \\
&=&\sum_{j\ge2}\langle\Phi(h_{T_j}),\Phi(h_{T_j})\rangle+2\sum_{2\le i<j}\langle\Phi(h_{T_i}),\Phi(h_{T_j})\rangle\\
&\le &(1+\delta_k)\sum_{i\ge 2}\|h_{T_i}\|_2^2+2\delta_{2k}\sum_{i>j\ge 2}\|h_{T_i}\|_2\|h_{T_j}\|_2\\
&\le & \sum_{i\ge 2}\|h_{T_i}\|_2^2+ \delta_{2k}(\sum_{i \ge 2}\|h_{T_i}\|_2)^2\\
&\le&
\left(\frac{(1-t)t^{(2-q)/q}}{k^{2/q-1}}+\frac{\delta_{2k}(1+(p_q-1)t^{1/q})^2}{k^{2/q-1}}\right)
\left(\sum_{i\ge 1}\|h_{T_i}\|_q^q\right)^{2/q}.
\end{eqnarray*}
where the last inequality follows from Lemma \ref{lem:3} and Lemma
\ref{lem:4}.
\end{proof}

\begin{lemma}[\cite{Lai11}]\label{lem:5}
For $q\in(0,1)$, we have
\begin{equation}
\|\Phi(h_{T_0}+h_{T_1})\|_2^2\ge \frac{ 1-\delta_{2k}}{k^{2/q-1}}(\tau(h,q)^2+t^{2/q})\left(\sum_{i\ge 1}\|h_{T_i}\|_q^q\right)^{2/q}.
\end{equation}
\end{lemma}

Combining Lemma \ref{lem:6} with Lemma \ref{lem:5}, we obtain
\[
(1-\delta_{2k}) (\tau(h,q)^2+t^{2/q})\le (1-t)t^{(2-q)/q}  +
\delta_{2k}(1+(p_q-1)t^{1/q})^2.
\]
After rearrangement of the terms, it implies that
\begin{equation}
\tau(h,q)^2\le (\delta_{2k}(1+(p_q-1)t^{1/q})^2
+t^{(2-q)/q}-(2-\delta_{2k})t^{2/q})/(1-\delta_{2k}). \label{tau:1}
\end{equation}
Then, an immediate sufficient condition for $x^q$ being the sparse
solution of (\ref{sp}) is to require the right hand side of
(\ref{tau:1}) being less than 1 for all $t\in(0,1]$. If we focus on
$\delta_{2k}<1$, the sufficient condition that the right hand side
of (\ref{tau:1}) being less than 1 can be expressed as
\begin{equation}
r(t,q,\delta_{2k}):=2\delta_{2k}
+(p_q-1)t^{1/q}(2+(p_q-1)t^{1/q})\delta_{2k}
+t^{(2-q)/q}-(2-\delta_{2k})t^{2/q} <1,~\forall t\in (0,1].
\label{tau:2}
\end{equation}
The problem then becomes to estimate the range of $q\in(0,1)$ and
$\delta_{2k}\in(0,1)$ for which (\ref{tau:2}) is true.

Notice that by setting $p_q=1,~\forall q\in(0,1)$ and by the fact
(shown in \cite{Lai11}) that when $q\rightarrow 0^+$,
\[
\sup_{t\in(0,1]}
\left\{t^{(2-q)/q}-(2-\delta_{2k})t^{2/q}\right\}\le q/e \rightarrow
0,
\]
one can immediately obtain Theorem \ref{thm:11}. In general, as
$p_q<1$ for $q\in(0,1)$, we expect an improvement over Theorem
\ref{thm:11}.

To begin with, we rewrite
\begin{eqnarray*}
r(t,q,\delta_{2k})=r_1(t,q,\delta_{2k})+r_2(t,q,\delta_{2k})+r_{3}(t,q,\delta_{2k}).
\end{eqnarray*}
where
\begin{eqnarray*}
r_1(t,q,\delta_{2k})&=&\delta_{2k} +\delta_{2k} (1-(1-p_q)t^{1/q})^2;\\
r_2(t,q,\delta_{2k})&=&t^{(2-q)/q};\\
r_3(t,q,\delta_{2k})&=& -(2-\delta_{2k})t^{2/q}
\end{eqnarray*}
with the ranges $\delta_{2k}, q\in(0,1)$ and $t\in(0,1]$.

We first discuss the monotonicity of functions $r,~r_1,~r_2,~r_3$.
Here is the summary:
\begin{enumerate}
\item[(a)] $r_{1}(t,q,\delta_{2k})$ and $r_{3}(t,q,\delta_{2k})$ are
decreasing functions of $t\in(0,1]$ whereas
$r_{2}(t,q,\delta_{2k})$ is increasing.
\item[(b)] $r_1(t,q,\delta_{2k})$ is decreasing in terms of $q$
since both $1-p_q$ and $t^{1/q}$ are increasing functions of
$q$.
\item[(c)] The sum of the latter two functions $(r_2+r_3)(t,q,\delta_{2k})$ is an increasing function of $q$
for $t\le\frac{1}{2-\delta_{2k}}$ since
\[
\frac{\partial(r_2+r_3)(t,q,\delta_{2k})}{\partial q}=-2/q^2  (\ln t)
t^{2/q-1}(1-(2-\delta_{2k})t)\ge 0.
\]
\item[(d)] $r(t,q,\delta_{2k})$ is an increasing function
of $\delta_{2k}$ since we can rewrite (\ref{tau:2}) as
\begin{eqnarray*}
r(t,q,\delta_{2k}) &=& \delta_{2k}(2
+2(p_q-1)t^{1/q}+(p_q-1)^2t^{2/q}+t^{2/q})+t^{(2-q)/q}-2 t^{2/q}\\
&=&\delta_{2k}(1+(1+(p_q-1)t^{1/q})^2)+t^{(2-q)/q}-2 t^{2/q}
\end{eqnarray*}
with $1+(1+(p_q-1)t^{1/q})^2>0$.
\end{enumerate}

%
%
Moreover, we have
\begin{lemma}\label{rt:lm}
Suppose $\delta_{2k}\le 1/2$, it holds that
\begin{equation*}
r(t,q,\delta_{2k})<1,~\forall t\in(\frac{1}{2-\delta_{2k}},1],~\forall
q\in(0,1).
\end{equation*}
\end{lemma}
\begin{proof}
Suppose $t> \frac{1}{2-\delta_{2k}}$, we always have
\begin{eqnarray}
r(t,q,\delta_{2k})&=&2\delta_{2k} +(p_q-1)t^{1/q}(2+(p_q-1)t^{1/q})
\delta_{2k} +(1/t-2+\delta_{2k})t^{2/q}\nonumber\\
&\le&2\delta_{2k} +(p_q-1)t^{1/q}(2+(p_q-1)t^{1/q})\delta_{2k}\nonumber\\
&<& 2\delta_{2k}\le 1. \nonumber
\end{eqnarray}
\end{proof}

Now we present our main result.
\begin{theorem} \label{thm:2}
Suppose $\delta_{2k}\le 1/2$. For any $q\in(0,0.9181]$, each
minimizer $x^q$ of the $\ell_q$ minimization (\ref{lq}) is the
sparse solution of (\ref{sp}).
\end{theorem}
\begin{proof}
Since $r(t,q,\delta_{2k})$ is an increasing function of
$\delta_{2k}$ by monotonicity (d), for any
$\delta'_{2k}<\delta_{2k}$, we have
\begin{equation*}
r(t,q,\delta_{2k})<1,~\forall t\in(0,1]
~\Longrightarrow~r(t,q,\delta'_{2k})<1,~\forall t\in(0,1].
\end{equation*}
Hence, it is sufficient to assume that $\delta_{2k}=1/2$.
Then we have
\[
  r(t,q,1/2) =1 +(p_q-1)t^{1/q}(2+(p_q-1)t^{1/q})/2
+t^{(2-q)/q}-3/2t^{2/q},
\]
which is less than 1, by Lemma \ref{rt:lm}, for all
$t\in(\frac{2}{3},1]$ and $q\in(0,1)$. The rest of the proof is to
show that $r(t,q,1/2)<1$ on $t\in(0,2/3]$ and $q\in(0,0.9181]$ by
incorporating monotonicity (a) - (c) on sufficiently fine meshes.

First, for any $q\in(0,1)$,
\begin{eqnarray*}
\frac{\partial r(t,q,1/2)}{\partial t}&=&\frac{t^{(2/q)-2}}{q}\left(
\frac{p_q-1}{t^{(1/q)-1}}+2-q-(3-(p_q-1)^2)t\right)\\
&\le&\frac{p_q-1}{q}t^{(1/q)-1}+\frac{2-q}{q}t^{(2/q)-2}\nonumber\\
&=&(p_q-1 +(2-q)t^{(1/q)-1})\frac{t^{(1/q)-1}}{q},\nonumber\\
&<&0,~~\forall ~0< t^{(1/q)-1}<\frac{1-p_q }{2-q}.
\end{eqnarray*}
In other words, $r(t,q,1/2)$ is strictly decreasing on $t\in(0,
(\frac{ 1-p_q }{2-q} )^{\frac{q}{1-q}})$. Moreover, $r(0,q,1/2)=1$.
Hence, for any $q\in(0,1)$, we have
\begin{equation}
r(t,q,1/2)<1, ~\forall~ 0<t< \left(\frac{ 1-p_q }{2
-q}\right)^{\frac{q}{1-q}},\label{eq:3}
\end{equation}
which specifies the range on which $r(t,q,1/2)<1$ by a function of
$q\in(0,1)$.

Secondly, to analyze the function $\left(\frac{ 1-p_q }{2
-q}\right)^{\frac{q}{1-q}}$, we can compute to get
\begin{equation*}
\lim\limits_{q\rightarrow0^+}\left(\frac{ 1-p_q }{2
-q}\right)^{\frac{q}{1-q}}=1;~\lim\limits_{q\rightarrow1^-}\left(\frac{ 1-p_q }{2
-q}\right)^{\frac{q}{1-q}}=0,
\end{equation*}
and its first derivative
\begin{eqnarray}
\frac{d}{dq}\left(\left(\frac{ 1-p_q }{2
-q}\right)^{\frac{q}{1-q}}\right)&=&\left(\frac{ 1-p_q }{2
-q}\right)^{\frac{q}{1-q}}\left( \frac{1}{(1-q)^2}\ln\frac{
  1-p_q }{2-q}+\frac{q(1-p_q-p_q'(2-q))}{ (1-q)(1-p_q)(2-q)}
  \right)\nonumber\\
  &=&\left(\frac{ 1-p_q }{2
-q}\right)^{\frac{q}{1-q}}\left( \frac{1}{(1-q)^2}\ln\frac{
  1-p_q }{2-q}+\frac{q(1-p_q-\frac{ 2\ln(q/2)}{(2-q)}p_q )}
  { (1-q)(1-p_q)(2-q)}
  \right) \label{eq:4}
\end{eqnarray}
where $(\ref{eq:2})$ is used for $p_q'$. We consider the following
two cases.
\begin{itemize}
\item[(i)]
Let $q\in(0,0.3]$. In this case, $p_q\ge p_{0.3} > 0.6081$.
According to Lemma \ref{lem:a1}, we have
\[
 \frac{-q\ln(q/2)}{(1-p_q)(2-q)}<1.
\]
Therefore,
\begin{eqnarray}
\frac{1}{(1-q)^2}\ln\frac{
  1-p_q }{2-q}+\frac{q(1-p_q-\frac{ 2\ln(q/2)}{(2-q)}p_q )}
  { (1-q)(1-p_q)(2-q)}
 & <&
\frac{1}{1-q}\left\{  \frac{\ln(\frac{
  1-p_q }{2-q})}{1-q}+
  \frac{q}{ 2-q}+
  \frac{ 2 p_q }{   2-q }\right\} \nonumber\\
  &\le&
  \frac{1}{1-q}\left\{\ln\big(\frac{
  1-0.6081 }{2-0.3}\big) +
  \frac{0.3}{ 2-0.3}+
  \frac{ 2 }{   2-0.3 }\right\}\nonumber\\
    &<&0,\nonumber
\end{eqnarray}
implying that $\left(\frac{ 1-p_q }{2
-q}\right)^{\frac{q}{1-q}}$ is strictly decreasing for
$q\in(0,0.3]$.
\item[(ii)]
Let $ q\in [0.3, 0.99].$ Define a partition of $q\in(0,0.3]$ by
$q_s=0.3+0.01s$ for $s=0,1,\cdots,69$ and try to estimate the
right hand side of (\ref{eq:4}) on each mesh $q\in[q_s,q_{s+1}]$
where $q_s, q_{s+1}$ are points in the partition. Then,
\begin{eqnarray}
 && \frac{\ln\frac{
  1-p_q }{2-q}}{(1-q)^2}+\frac{q(1-p_q-\frac{ 2\ln(q/2)}{(2-q)}p_q )}{
(1-q)(1-p_q)(2-q)}\nonumber\\
&=&
  \frac{\ln\frac{
  1-p_q }{2-q}}{(1-q)^2}+\frac{q }{
(1-q) (2-q)}+\frac{-2q\ln(q/2)  p_q }{
(1-q)(1-p_q)(2-q)^2}\nonumber\\
&\le&
  \frac{\ln\frac{
  1-p_{q_{s+1}} }{2-q_{s+1}}}{(1-q_{s})^2}+\frac{q_{s+1} }{
(1-q_{s+1}) (2-q_{s+1})}+\frac{-2q_{s+1}\ln(q_{s}/2)  p_{q_{s}} }{
(1-q_{s+1})(1-p_{q_{s}})(2-q_{s+1})^2}.\label{case ii}
\end{eqnarray}
With the aid of computer, the evaluation of (\ref{case ii})
shows that they are negative on all mesh points. That is,
$\left(\frac{ 1-p_q }{2 -q}\right)^{\frac{q}{1-q}}$ is strictly
decreasing for $q\in [0.3, 0.99]$.
\end{itemize}
Together with (i) and (ii), we conclude that $\left(\frac{ 1-p_q }{2
-q}\right)^{\frac{q}{1-q}}$ is strictly decreasing for $q\in (0,
0.99]$, which leads to the following estimation:
\begin{eqnarray}
\left(\frac{ 1-p_q }{2 -q}\right)^{\frac{q}{1-q}}\ge \left(\frac{
1-p_{0.9181}}{2
-0.9181}\right)^{\frac{0.9181}{1-0.9181}}>0.0105,&&\forall
q\in (0, 0.9181];\label{0.9181}\\
\left(\frac{ 1-p_q }{2 -q}\right)^{\frac{q}{1-q}}\ge \left(\frac{
1-p_{0.17} }{2 -0.17}\right)^{\frac{0.17}{1-0.17}}>
0.6821,&&\forall q\in(0,0.17].\label{0.17}
\end{eqnarray}
It follows from (\ref{0.9181}) and (\ref{eq:3}) that
\begin{eqnarray*}
r(t,q,1/2)<1,&&\forall  t\in (0, 0.0105], ~\forall q\in(0,0.9181];
\end{eqnarray*}
from (\ref{0.17}) and Lemma \ref{rt:lm} that
\begin{eqnarray*}
r(t,q,1/2)<1,&&\forall t\in(0,1],~\forall q\in(0,0.17].
\end{eqnarray*}
That is, to prove (\ref{tau:2}), it is left to verify that
$r(t,q,1/2)< 1$ on $(t,q)\in[0.0105,2/3]\times[0.17,0.9181]$. Our
idea is to subdivide the region $[0.0105,2/3]\times[0.17,0.9181]$
into the union of small squares with the type
$[t_r,t_{r+1}]\times[q_s,q_{s+1}]$. Then, by the monotonicity (a),
(b) and (c), we can estimate $r(t,q,1/2)$ on this square by the
corner points as follows:
\begin{eqnarray}
&&\max_{t\in[t_r,t_{r+1}]}\max_{q\in[q_s,q_{s+1}]}r(t,q,1/2)\nonumber\\
&\le&r_{1}(t_r,q_s,1/2)+r_{2}(t_{r+1},q_{s+1},1/2)+
r_{3}(t_r,q_{s+1},1/2).\label{mesh-est}
\end{eqnarray}
If the evaluation by computer shows that (\ref{mesh-est}) is less
than one, we are done with the square. Otherwise, the estimation
might not be tight enough so that we have to subdivide the square
into finer meshes.

Our calculation shows that, by covering
$[0.0105,2/3]\times[0.17,0.9181]$ with
$$[0.0105,0.6667]\times[0.17,0.9172]\bigcup
[0.0105,0.66667]\times[0.9172,0.91809]\bigcup
[0.0105,0.666667]\times[0.91809,0.9181],
$$
we can get the desired result as
\begin{eqnarray}
&&\max_{r\in\{105,106,\cdots,6666\}}\{
\max_{s\in\{1700,1701,\cdots,9171\}} \left\{r_{1}(t_r,q_s,1/2)
\right. \nonumber\\
&& ~~~~~~~~~~~~~~~~~~~ \left.+r_{2}(t_{r+1},q_{s+1},1/2)
+r_{3}(t_{r+1},q_{s+1},1/2)\right\} \} < 1, \nonumber
\end{eqnarray}
where $t_r=0.0001r$ and $q_s=0.0001s$;
\begin{eqnarray}
&&\max_{r\in\{1050,1051,\cdots,66666\}}  \{
\max_{s\in\{91720,91721,\cdots,91808\}} \left\{r_{1}(t_r,q_s,1/2) \right. \nonumber\\
&& ~~~~~~~~~~~~~~~~~~~ \left.+r_{2}(t_{r+1},q_{s+1},1/2)
+r_{3}(t_{r+1},q_{s+1},1/2)\right\} \} < 1,  \nonumber
\end{eqnarray}
where $t_r=0.00001r$ and $q_s=0.00001s$ in this partition; and
\begin{eqnarray}
&& \max_{r\in\{10500,10501,\cdots,666666\}}
\left\{r_{1}(t_l,q_{91809},1/2) \right.
\nonumber\\
&& ~~~~~~~~~~~~~~~~~~~ \left. +r_{2}(t_{r+1},q_{91810},1/2)
+r_{3}(t_{r},q_{91810},1/2)\right\} < 1.  \nonumber
\end{eqnarray}
where $t_r=0.000001r$.

All the above calculations were carried out by computer and the
proof is thus complete.
\end{proof}

\begin{remark}
Let $q\rightarrow 0^+$. Then $p_q\rightarrow 1$ and it follows from
(\ref{tau:2}) that $\delta_{2k}\le \frac{1}{2}$. In other words,
$\delta_{2k}\le \frac{1}{2}$ can not be further improved based on
the sufficient condition (\ref{tau:1}). Furthermore, one can verify
that
\[
r(0.064,0.9182,1/2)>1.0000002,
\]
implying that the condition $q\in(0,0.9181]$ in Theorem  \ref{thm:2}
is tight up to the fourth decimal digit.
\end{remark}

\begin{remark}
Based on the analysis in \cite{Lai11}, the threshold $q_0$ in
Theorem \ref{thm:11} can be estimated to be around $0.0513$, far
less than the number 0.9181 reported in our Theorem \ref{thm:2}.
\end{remark}

From the above analysis, if the sparse recovery is to be exact for
any $q\in(0,1)$, we may require a tighter restricted isometric
constant than $\delta_{2k}=0.5$. This is the spirit of Theorem
\ref{thm:09}, which we shall show immediately an improvement on
their result.

To investigate the issue, we look into the case when $q=1, p_q=1/4$.
The sufficient condition (\ref{tau:2}) becomes
\begin{equation}
r(t,1,\delta_{2k}) =2\delta_{2k} +(1-3 \delta_{2k}/2)t
-(2-25\delta_{2k}/16)t^{2} <1,~\forall t\in (0,1].\label{q=1}
\end{equation}
Since $2-25\delta_{2k}/16>0$, $r(t,1,\delta_{2k})$ is a concave
parabola of $t\in (0,1]$. The sufficient condition in (\ref{q=1})
holds if and only if the equation $r(t,1,\delta_{2k})=1$ has no
solution. Namely, we need
%
%
\[
(1-3\delta_{2k}/2)^2+4(2-25\delta_{2k}/16)(2\delta_{2k}-1)<0.
\]
It follows that
\[
\delta_{2k}< \frac{77-\sqrt{1337}}{82}~\left(>
0.493109\right).
\]

We therefore have the following theorem.
\begin{theorem}
Suppose $\delta_{2k}\le 0.4931$. Then for any $q\in(0,1)$ each
minimizer $x^q$ of the $\ell_q$ minimization (\ref{lq}) is the
sparse solution of (\ref{sp}).
\end{theorem}
\begin{proof}
Since $r(t,q,\delta_{2k})$ is an increasing function of
$\delta_{2k}$, we fix $\delta_{2k}$ at $0.4931$.
According to Lemma \ref{rt:lm}, $r(t,q,0.4931)<1$ for $t\in
(0.663614,1]$. Therefore, it is sufficient to check the maximum of
$r(t,q,0.4931)$ over the region $(t,q)\in [0,0.6637]\times
[0.9181,1]$. This is done by dividing the region into two parts
$$[0,0.6637]\times[0.9181,1]=[0,0.6637]\times[0.9181,0.9992]\bigcup[0,0.6637]\times
[0.9992,1],$$ the first of which with a mesh size 0.0001 whereas the
latter with 0.00001. The computation has been verified by computer
as follows
\begin{eqnarray*}
&&\max_{ t\in[0,0.6637]}\max_{q\in[0.9181,1]}r(t,q,0.4931)\\
&\le&\max \{ \{
\max_{r=0,1,2,\cdots,6636} \max_{s=9181,9182,\cdots,9991}
\left\{r_{1}(t_r,q_s,0.4931) \right.\\
&&\left.~~~~~~~~~~~~+r_{2}(t_{r+1},q_{s+1},0.4931)
+r_{3}(t_{r},q_{s+1},0.4931)\right\} \};\\
&&\max_{u=0,1,2,\cdots,66369} \{\max_{v=99920,99921,\cdots,99999}
\left\{r_{1}(t_u,q_v,0.4931) \right.\\
&&\left.~~~~~~~~~~~~+r_{2}(t_{u+l},q_{v+1},0.4931)
+r_{3}(t_{u},q_{v+1},0.4931)\right\} \} \}\\
&<&1
\end{eqnarray*}
where $t_r=0.0001r, q_s=0.0001s, t_u=0.00001u, q_v=0.00001v$ and the
proof is thus complete.
\end{proof}

\section{Conclusion}
In this paper, based on the extended ``converse of a generalized
Cauchy-Schwarz inequality'' of quasi-norm,
 we establish new sufficient conditions under
which the minimizer of the $\ell_q$ minimization is the sparse
solution of the corresponding underdetermined linear system. More
precisely, we show that if the restricted isometry constant
$\delta_{2k}\le 1/2$,
then for any $q\le 0.9181$ the solution of $\ell_q$ minimization
also solves the $\ell_0$ minimization. Furthermore, if
$\delta_{2k}\le 0.4931$, then for any $q\in(0,1]$, the minimizer of
$\ell_q$ minimization remains optimal in $\ell_0$ minimization. Our
results strongly improves those reported previously in the
literature. We believe the value $0.4931$ can be improved to
$\frac{77-\sqrt{1337}}{82}$ with a new proof, but now it is still
open. The other future research is to extend the presented results
to the
 rank minimization problem.

\section*{Acknowledgments}
This research was undertaken while Y. Hsia visited National Cheng
Kung University, Tainan, Taiwan.

\end{document}